\documentclass[conference]{IEEEtran}

\usepackage{tikzsymbols}
\usepackage{fancyhdr}
\usepackage{lastpage}

\usepackage{caption} 
\usepackage{subcaption}
\usepackage{times}
\usepackage{amsmath, amsthm, amssymb}
\usepackage{color}
\usepackage{graphicx}
\usepackage{multirow}
\usepackage{cite}
\usepackage{epsfig}
\usepackage{url}
\usepackage{enumitem}
\usepackage{bm}
\usepackage{sansmath}
\usepackage{thmtools,thm-restate}
\usepackage{xspace}

\definecolor{light-gray}{gray}{0.95}
\definecolor{gray}{gray}{0.90}

\renewcommand{\baselinestretch}{.97}
\newcommand{\comment}[1]{}
\newcommand{\defn}[1]{\textbf{\emph{#1}}}
\newcommand{\poly}[1]{{poly({#1})}}

\newtheorem{theorem}{Theorem}
\newtheorem{lemma}{Lemma}

\newtheorem{corollary}{Corollary}

\newcommand{\suc}{\mbox{suc}}

\newcommand{\spath}{search path\xspace}

\newcommand{\whp}{w.h.p\xspace}

\newcommand{\qm}{\mbox{group}\xspace}
\newcommand{\qms}{\mbox{groups}\xspace}
\newcommand{\Qm}{\mbox{Group}\xspace}
\newcommand{\Qms}{\mbox{Groups}\xspace}

\newcommand{\Q}{\mbox{G}}
\newcommand{\G}{\mbox{H}} 
\newcommand{\QG}{\mathcal{G}} 

\newcommand{\pbad}{p_f}
\newcommand{\psearch}{q_f}

\newif\ifcomments
\commentstrue

\IEEEoverridecommandlockouts
\begin{document}

\title{\huge Tiny Groups Tackle Byzantine Adversaries\vspace{-8pt}}

\author{\IEEEauthorblockN{Mercy O. Jaiyeola\IEEEauthorrefmark{1},
Kyle Patron\IEEEauthorrefmark{2},
Jared Saia\IEEEauthorrefmark{3}, 
Maxwell Young\IEEEauthorrefmark{1},
Qian M. Zhou\IEEEauthorrefmark{1}\mbox{\hspace{2.6cm}}}\vspace{-8pt}
\and
\IEEEauthorblockA{\IEEEauthorrefmark{1}Mississippi State University,\\
Dept. of Computer Science and Eng.,\\
Miss. State, MS, USA\\ {\small \texttt{\{moj31,my325,qz70\}@msstate.edu}}}
\and
\IEEEauthorblockA{\IEEEauthorrefmark{2}Palantir Technologies\\ 
New York, NY, USA\\
{\small \texttt{kyle.patron@gmail.com}}}
\and
\IEEEauthorblockA{\IEEEauthorrefmark{3}University of New Mexico,\\
Department of Computer Science,\\
Albuquerque, NM, USA\\ {\small \texttt{saia@cs.unm.edu}}}}

\maketitle

\setlength{\footskip}{15pt}
\thispagestyle{plain}
\pagestyle{plain}
 \setcounter{page}{1} 

\begin{abstract}
A popular technique for tolerating malicious faults in open distributed systems is to establish small \defn{\qms} of participants, each of which  has a non-faulty majority.  These \qms are used as building blocks to design attack-resistant algorithms.

Despite over a decade of active research, current  constructions require \qm~sizes of {\boldmath $O(\log n)$}, where {\boldmath $n$} is the number of participants in the system. This \qm~size is important since communication and state  costs scale polynomially with this parameter. Given the stubbornness of this logarithmic barrier, a natural question is whether better bounds are possible.

Here, we consider an attacker that controls a constant fraction of the total computational resources in the system. By leveraging proof-of-work (PoW), we demonstrate how to reduce the \qm~size \defn{exponentially} to {\boldmath $O(\log\log n)$} while maintaining strong security guarantees.  This reduction in \qm~size yields a significant improvement in  communication and state costs.\end{abstract}


\IEEEpeerreviewmaketitle

\section{Introduction}\label{sec:intro}\vspace{-0pt}
Byzantine fault tolerance addresses the challenge of performing useful work when participants in the system are malicious. Participants, or \defn{identifiers (IDs)} in the system, may be malicious; these malicious IDs can discard or corrupt information that is routed through them or stored on them.   


A popular technique for overcoming these challenges is to arrange IDs into sets called \defn{\qms},\footnote{Such sets have appeared under different names in the literature, such as ``swarms''~\cite{fiat:making}, ``clusters''~\cite{guerraoui:highly}, and ``quorums''~\cite{young:towards}. Our choice of ``groups'' aligns with pioneering work in this area~\cite{awerbuch_scheideler:group}.} where each has a non-faulty majority.  We can then ensure the following:\vspace{4pt}
\begin{itemize}

\item Secure routing is possible. For \qms $\Q_1$ and $\Q_2$ along a route,  all members of $\Q_1$ transmit messages to all members of $\Q_2$. This all-to-all exchange, followed by majority filtering by each non-faulty ID in $\Q_2$,  guarantees correctness of communication between \qms despite malicious IDs.\vspace{5pt}

\item Computation is performed by all members of a \qm~via protocols for \defn{Byzantine agreement (BA)}~\cite{lamport1982byzantine},  or more general secure multiparty computation~\cite{Yao:1986:GES:1382439.1382944}, to guarantee that tasks execute correctly.  In this way, each \qm~simulates a \defn{reliable processor} upon which jobs can be run.

\vspace{3pt}

\end{itemize}

The use of \qms provides a {\it scalable} approach to designing an attack-resistant distributed system, by avoiding the need to have all $n$ IDs perform BA in concert.\footnote{\Qms also improve robustness in other ways. Members may agree to ignore an ID if it misbehaves too often, hence reducing spamming. Data may also be redundantly stored at multiple \qm~members.}\smallskip\smallskip


\noindent Designing attack-resistant systems using \qms has been an active topic of research for over a decade, with many results~\cite{HK,fiat:making,awerbuch:towards,awerbuch:random,awerbuch:towards2,saia:reducing,castro:secure,halo:kapadia,salsa:nambiar,young:towards,guerraoui:highly,naor:novel,saad:self-healing}. Yet, despite this progress, an enduring requirement is that each \qm~contains $O(\log n)$ members; that is, the group size grows logarithmically in $n$.

Why does this logarithmic size matter? At first glance, it is an unlikely bottleneck. However, since \qms~are building blocks for the system, their size, $|\Q|$, impacts costs: \vspace{2pt}
\begin{enumerate}[label=(\roman*)]

\item{\it Cost of \Qm~Communication.}  \Qm~members must often act in concert; for example, executing distributed key generation~\cite{young:towards}, or generating random numbers~\cite{fiat:making,awerbuch:random}.  Such protocols require messages be exchanged between all members.  We label this as \defn{group communication} and it has $\Theta(|\Q|^2)$ message cost. \vspace{4pt}

\item{\it Cost of Secure Routing.} Routing via all-to-all exchange between two \qms incurs $\Omega(|\Q|^2)$ message complexity. Given a route of length $D$, communication between any two \qms~requires $O(D|\Q|^2)$ messages.\footnote{Improvements are possible, but they come with caveats.  Results in \cite{fiat:making,saia:reducing} lower the cost to $O(D|\Q|)$ in expectation but require a non-trivial (expander-like) construction,  and~\cite{young:towards} further reduces this to $O(D)$ in expectation but with a \poly{$|\Q|$} message cost each time routing tables are updated which is expensive even with moderate churn.}\vspace{5pt}

\item{\it  Cost of State Maintenance.} Each ID $w$ must maintain state on  its neighbors; this includes both the members of all \qms to which $w$  belongs {\it\underline{and}} the members of neighboring \qms. This requires storing link information, as well as periodically testing links for liveness.\footnote{Constructions in~\cite{naor:novel,fiat:making} have each ID belonging to $\eta>1$ \qms for a state overhead of $\Omega(|\Q| \eta )$. Also, if each \qm~links to $\Delta$ neighboring \qms, then $O(|\Q|\Delta)$ links must be maintained; typically, $\Delta = O(\log n)$.}\vspace{4pt}
\end{enumerate}

\noindent In each case above, reducing $|\Q|$ would directly reduce cost. Unfortunately, in prior results,  $|\Q| \approx \log n$  is key to ensuring  that all \qms have a non-faulty majority with high probability (w.h.p.).\footnote{With probability at least $1-1/n^{c'}$ for a tunable constant $c' \geq 1$.} Without this property, all previous \qm~constructions succumb to adversarial attack. A natural question is: {\it Are there new ideas that allow us to decrease $|\Q|$ while maintaining strong security guarantees?}\smallskip

In response, we consider an attacker that controls a constant fraction of the total computational resources in the system. By employing proof-of-work (PoW), we obtain our main result: {\bf \Qm~size can be reduced exponentially while still allowing routing in all but a vanishingly small portion (an {\boldmath{$o(1)$}}-fraction) of the network}.\smallskip 


\subsection{Defining {\large\boldmath{$\varepsilon$}}-Robustness}  

Consider a system of $n$ IDs and $n$ \qms, where a $\beta$-fraction of IDs are malicious; such IDs may deviate arbitrarily from any prescribed protocol to derail operations in the network.  The following defines our notion of  \defn{\boldsymbol{$\varepsilon$}-robustness}:\smallskip

\noindent {\it  For a small $\varepsilon>0$, at least $(1-\varepsilon)n$ \qms have a non-faulty majority and can securely route messages to each other.}\smallskip



\noindent The parameter $\beta$ is a sufficiently small positive constant  less than $1/2$, and $\varepsilon=o(1)$.  We consider the following questions:\smallskip

\noindent{\bf Is this a useful concept?} Consider decentralized storage and retrieval of data. This definition guarantees all but an $\varepsilon$-fraction of data is reachable and maintained reliably.  Example applications include distributed databases, name services, and content-sharing networks. Alternatively, consider $n$ jobs in an open computing platform that are run on individual machines. This definition guarantees that all but an $\varepsilon$-fraction of those jobs can be correctly computed.\footnote{While this may not be sufficient for general computation, it is valuable for tasks where an $o(1)$ error rate or bias can be tolerated; for example, obtaining statistics on a group of machines to measure network performance.}

\smallskip 


\noindent{\bf Is satisfying this definition trivial?} Given $\Theta(n)$ non-faulty IDs, this definition characterizes simulating $(1-\varepsilon)n$ reliable processors \underline{\it and} the ability to route information between them.  If we ignore the use of \qms or, equivalently, consider \qms each consisting of a single ID, then we trivially have $(1-\beta)n$ reliable processors, but routing between them is challenging. For example, establishing links between all pairs of IDs will give secure routing, but this is hardly scalable.  \smallskip 

\noindent{\bf Do previous solutions solve this problem?} Prior results using \qms are subsumed by this definition and they address  $\varepsilon = 1/$\poly{$n$}. In this case,  routing is possible -- albeit, costly -- because w.h.p. {\it all} \qms have a non-faulty majority.

To reduce cost, we consider {{\boldmath{$\varepsilon = 1/${\bf\poly{$\log n$}}}}. We show how this allows us to reduce the \qm~size exponentially,  but at the price of having a small fraction of the \qms with a majority of malicious IDs.  \vspace{-2pt}


\subsection{Related Work}\label{sec:related-work}

\noindent{\bf Tolerating Byzantine Faults via \Qms.} The use of \qms for building attack-resistant distributed systems has received significant attention, and all address the case of $\varepsilon=1$/\poly{$n$}. 

Early results obtain a \poly{$\log n$} factor increase in costs, assuming constraints on the amount of system dynamism~\cite{HK,naor_wieder:a_simple,fiat:making,fiat_saia:censorship,awerbuch_scheideler:group}. 

Full dynamism was achieved by Awerbuch and Scheideler in a series of breakthrough results~\cite{awerbuch:towards,awerbuch:random,awerbuch:towards2}; they propose a {\it cuckoo rule} that w.h.p. preserves a non-faulty majority in all \qms over $n^{\Theta(1)}$ joins/departures when the system size remains $\Theta(n)$. More recently, Guerraoui~et al.~\cite{guerraoui:highly} give similar guarantees when the system size can vary polynomially. 

Simulations of the cuckoo rule are conducted in~\cite{sen:commensal}.  The trade-off between \qm~size and security is examined, and findings suggest that $|\Q|$ must be fairly large. For $n=8,192$ (the largest size examined) and $\beta \approx 0.002$,  $|\Q|=64$ preserves a non-faulty majority in each \qm~for $10^5$  joins/departures; $\beta \approx 0.07$ is possible with suggested improvements in~\cite{sen:commensal}. 

Several results have focused on reducing communication costs when the good majority of all \qms is guaranteed via an algorithm like the cuckoo rule~\cite{saia:reducing,young:towards}. Here too, \qm~size impacts performance and $|\Q|=30$ incurs significant latency in PlanetLab experiments~\cite{young:towards}.  \Qms have also been used in conjunction with quarantining malicious IDs~\cite{saad:self-healing2,saad:self-healing} with limited churn. 

Finally, we observe that none of these results explicitly uses PoW, with the possible exception of~\cite{awerbuch_scheideler:group}, where  computational challenges or Turing tests are briefly discussed as a means for throttling the join rate of Byzantine IDs. These prior results assume a model where the fraction of Byzantine IDs is always limited to strictly less than $1/2$. In contrast, the use of PoW provides a plausible mechanism by which to enforce this limit (see below for more discussion).\smallskip



\noindent{\bf Attack-Resistance Without \Qms.} Other  distributed  constructions exist that do not explicitly use \qms~\cite{saia_fiat_gribble_karlin_saroiu:dynamically,fiat_saia:censorship,datar:butterflies}. However, the associated techniques retain some form of $O(\log n)$ redundancy with regards to data placement or route selection and, therefore, incur the typical \poly{$\log n$} cost.  

In~\cite{castro:secure,halo:kapadia,salsa:nambiar}, malicious faults are tolerated by routing along multiple diverse routes. However, it is unclear that these systems can provide theoretical guarantees on robustness.


Byzantine resistance when $O(\sqrt{n}/$\poly{$\log n$}$)$ IDs may depart and join {\it per time step} is examined in~\cite{augustine:fast,Augustine2015}.\footnote{High churn without Byzantine fault tolerance is considered in \cite{7354403,Augustine:2012,Augustine:2013:SSD:2486159.2486170}.} In this challenging model, (roughly) $O(\sqrt{n})$ Byzantine IDs can be tolerated. Our result addresses more moderate churn while tolerating $\Theta(n)$ Byzantine IDs.

Central authorities (CAs) have been used in  prior  results\cite{castro:secure,rodrigues:rosebud} to achieve robustness. While our results can be used in conjunction with a CA, it is not always plausible to assume such an authority is available and immune to attack. For this reason, our work does not depend on a CA.\smallskip



\noindent{\bf Computational Puzzles.} Proof-of-work (PoW) via computational puzzles has been used to mitigate the  Sybil attack~\cite{douceur02sybil}, whereby an adversary overwhelms a system with a large number of malicious IDs. We note that such PoW schemes have been proposed in decentralized settings; for examples, see~\cite{li:sybilcontrol,pow-without}. However, such PoW schemes only {\it limit the number of Sybil IDs} --- typically commensurate with the amount of computational power available to the adversary --- and the problem of tolerating these adversarial IDs must still be addressed by other means; for examples, see~\cite{yu:survey,scheideler:shell}.


A prominent example of using PoW to provide security is Bitcoin~\cite{nakamoto:bitcoin}. However, note that analyses of Bitcoin and related systems typically assume a communication primitive that allows an ID to disseminate a value to all other IDs within a known bounded constant amount of time despite an adversary~\cite{Garay2015,Micali16,Luu:2016}. In contrast, our results do not assume the existence of such a primitive. 

We note that PoW imposes a computational overhead on the system participants. Nonetheless, examples such as Bitcoin and emerging blockchain technologies (for example, Ethereum~\cite{ethereum}) illustrate the success of PoW in practice, and exemplify that computational overhead from PoW may be tolerable given the security guarantees received in exchange.



\subsection{Our Model and Preliminaries}\label{sec:model}\vspace{-2pt}

\noindent We consider a system of $n$ IDs. An ID is \defn{good} if it obeys the protocol; otherwise, the ID is Byzantine or \defn{bad}. \smallskip

\noindent{\bf The Adversary.} We assume that our adversary controls a $\beta$-fraction of the computational power in the network, where $\beta$ is a sufficiently small positive constant less than $1/2$.\footnote{For simplicity in our proofs, we let $\beta$ be small since we do not try to optimize this quantity. However, larger values of $\beta$ are likely possible.}  This is a common assumption when using PoW to design attack-resistant, open systems~\cite{parno:portcullis,li:sybilcontrol,pow-without}.  For simplicity, throughout Sections~\ref{sec:static} and~\ref{sec:dynamic}, we assume there is always at most a $\beta$ fraction of bad IDs. In Section~\ref{sec:pow}, this assumption is justified by proving that the adversary is constrained, via proof-of-work assumptions, to generating (roughly) at most a $\beta$-fraction of IDs at any time.  

We assume that a single adversary controls all the bad IDs.  This is a challenging model since a single adversary allows the bad IDs to perfectly collude and coordinate their attacks. The adversary also knows the network topology and all message contents; however, the adversary does not know the random bits generated locally by any good ID.\vspace{3pt}

\noindent{\bf \Qms.}  In our system, each \qm~has size $\Theta(\log\log n)$. Each ID $w$ has its own \qm~$\Q_w$ and $w$ is referred to as the \defn{leader}.  A \qm~$\Q$ is \defn{good} if (i) $d_1\ln\ln n \leq |\Q| \leq d_2\ln\ln n$ for appropriate constants $d_1 < d_2 $, and (ii) the number of bad IDs in $\Q$ is at most $(1+\delta)\beta|\Q|$ for some tunably small constant $\delta>0$ depending only on sufficiently large $n$; otherwise, the group is \defn{bad}. Note that \qms are not necessarily disjoint;  in addition to being the leader of $\Q_w$, ID $w$ may belong to other \qms. \Qm~construction is described in Section~\ref{sec:join-dep}. \vspace{3pt}

\noindent{\bf Input Graph.}  Our result builds off an \defn{input graph} {\boldmath $H$} on {\boldmath $N$} vertices, where each vertex corresponds to an ID.\footnote{For clarity, we use $N$ to denote the number of IDs in $\G$, and we use $n$ for the number of IDs in our attack-resistant construction.}  Each ID is a virtual participant, and each ID is represented as  a value in the interval $[0,1)$ known as the \defn{ID space}; this is viewed as a unit ring where moving \defn{clockwise} along the ring corresponds to moving from away from $0$ towards $1$. The \defn{successor} of a point $x$  in $[0,1)$ is the first ID encountered by moving clockwise from $x$ on the unit ring; this is denoted by \defn{\suc}{\boldmath{($x$)}}.  Assuming there is no adversary and that IDs are distributed independently and uniformly at random (u.a.r.) in $[0,1)$, the following properties hold for $\G$ with probability at least $1 - N^{-c'}$ for a tunable constant $c'\geq1$. \vspace{3pt}

\begin{description}[leftmargin=5pt]
\item{\it P1 -- Search Functionality.} There exists a \defn{search} (or \defn{routing}) \defn{algorithm} that, for any \defn{key value} in $[0,1)$, returns contact information for the ID responsible for the corresponding \defn{resource} (i.e., data item, computational job, network printer, etc.). A search requires traversing {\boldmath{$D$}} $= O(\log N)$ IDs.\footnote{An ID is ``traversed''  if the ID lies on the path of search between the initiator of the search and the ID responsible for the resource being sought. For more discussion, please see our online version~\cite{JaiyeolaPSYZ17},  Appendix~\ref{sec:background}.} \vspace{3pt}
\item{\it P2 -- Load Balancing.} A randomly chosen ID is responsible for at most a $(1+\delta'')/N$-fraction of the key values (and the corresponding resources) for an arbitrarily small $\delta''>0$ depending on sufficiently large $N$. \vspace{3pt}
\item{\it P3 -- Linking Rules.} Each ID $w$ links to IDs in a set of neighbors {\boldmath{$S_w$}}; $|S_w|=O(\log^{\gamma} n)$ for some constant $\gamma>0$.  Any ID may determine the elements in $S_w$ by performing searches.\footnote{For example, in Chord~\cite{stoica_etal:chord}, $w$'s neighbors are (1) the ID counter-clockwise and the ID clockwise to $w$ on the unit ring, and (2) the respective successors of  points $w+\Delta(i)$, where $\Delta(i)$ is an exponentially increasing distance in the ID space for  $i=1, ..., O(\lg N)$. Any ID may verify that $u$ is a neighbor of $w$ via a search on $w+\Delta(i)$ and checking that the result is $u$.} \\
\mbox{~~~~~}There are also $O(\textup{poly}(\log n))$ IDs each of which has $w$ in its respective set of neighbors. Again, any ID may verify this by performing searches. The number of links on which ID $w$ is incident is the \defn{degree} of $w$, and every ID has the same degree asymptotically.\vspace{3pt}


\item {\it P4 -- Congestion Bound.} The \defn{congestion} is $C = O(\log^{c}n/n)$ for a constant $c \geq 0$, where congestion is the maximum probability (over all IDs) that a ID is traversed in a search initiated at a randomly chosen ID for a randomly chosen point in $[0,1)$.\vspace{2pt}
\end{description}

\noindent{\defn{We emphasize that $\G$ is not assumed to provide any security or performance guarantees if there are bad IDs present. Rather, any such $\G$ provides a viable topology that, using our result, can be made to tolerate bad IDs.} Note that many constructions for $\G$ exist such as Chord~\cite{stoica_etal:chord},  the distance-halving construction~\cite{naor:novel}, Viceroy~\cite{malkhi_naor_ratajczak:viceroy}, Chord++~\cite{chordplus:awerbuch}, D2B~\cite{fraigniaud:d2b}, FISSIONE~\cite{li:fissione}, and Tapestry~\cite{zhao_kubiatowicz_joseph:tapestry}.  


\smallskip\smallskip

\noindent{\bf IDs and PoW in Our Construction.} An ID is a virtual participant in the network, and each ID is represented as a value in $[0,1)$ in our construction. Note that adequate precision is obtained using $O(\log n)$ bits. Important properties that our system guarantees are: \vspace{0pt}
\begin{itemize}[leftmargin=12pt]
\item IDs expire after a period of time that can be set by the system designers.
\item A claim to own an ID can be verified by any good ID.
\item The adversary possesses (roughly) at most $\beta n$ IDs, and these IDs are u.a.r. from the ID space $[0,1)$.\vspace{0pt}
\end{itemize}

\noindent We emphasize that our construction does not take these properties for granted; rather they are enforced via a PoW scheme.  However, given space constraints and that the bulk of our results are proved without the need to reference these details, we will assume these properties in Sections~\ref{sec:static} and~\ref{sec:dynamic}. We remove these assumptions in Section~\ref{sec:pow}.

We make the \emph{random oracle assumption}~\cite{bellare1993random}: there exist hash functions, $h$, such that $h(x)$ is uniformly distributed over $h$'s range, when any $x$ in the domain of $h$ is input to $h$ for the first time. We assume that both the input and output domains are the real numbers $[0,1)$. In practice, $h$ may be a cryptographic hash function, such as SHA-2~\cite{sha2}, with inputs and outputs of sufficiently large bit lengths.\smallskip

\noindent{\bf Joins and Departures.} Our work addresses a dynamic system where IDs may join and depart. We delay our description of this aspect until Section~\ref{sec:dynamic}.\smallskip

\noindent We use the following concentration results.\vspace{-0pt}

\begin{theorem}\label{thm:Chernoff}
(Chernoff Bounds~\cite{motwani_raghavan:randomized}) Let $X_1,\dots, X_N$ be independent indicator random  variables such that $\Pr(X_i) = p$ and let $X = \sum_{i=1}^{N} X_i$.  For any $\delta$, where $0 < \delta < 1$, the following holds:\vspace{-8pt}

$$\Pr(X>(1+\delta)\,E[X]) \leq e^{-\delta^2 \,E[X] / 3}\vspace{-3pt}$$ 
$$\Pr(X<(1-\delta)\,E[X]) \leq e^{-\delta^2 \,E[X] / 2}$$
\end{theorem}

\begin{theorem}\label{thm:MOBD}
(Method of Bounded Differences, Corollary 5.2 in~\cite{dubhashi:concentration})  Let $f$ be a function of the variables $X_1, ..., X_N$ such that for any $b,b'$ it holds that $|f(X_1, ..., X_i = b, ..., X_N ) - f(X_1, ..., X_i = b', ..., X_N)| \leq c_i$ for $i=1, ..., N$. Then, the following holds:\vspace{-5pt}
$$Pr( f > E[f] + t)  \leq  e^{-2t^2/(\sum_{i} c_i^2)}\vspace{-3pt}$$ 
$$Pr( f < E[f] - t) \leq e^{-2t^2/(\sum_{i} c_i^2)}\vspace{-0pt}$$
\end{theorem}
\vspace{-3pt}
\noindent{}Finally, all of our results hold given that $n$ is sufficiently large; we assume this throughout.\vspace{-3pt}


\subsection{Overview and Our Main Result}\label{sec:overview-main}\vspace{-1pt}

As discussed above, reducing \qm~size is desirable, but gives rise to the possibility of bad \qms.  In Section~\ref{sec:static}, we demonstrate how to achieve $1/\textup{poly}{(\log n)}$-robustness with \qms of size $\Theta(\log\log n)$ when there is no churn. This argument leverages the bound on congestion given by the input graph, along with a careful tallying of the fraction of ID space that cannot be securely searched. 

This result is applied in Section~\ref{sec:dynamic} where we show that  $O(1/\textup{poly}{(\log n)})$-robustness can be maintained with churn. A key component of our construction is the use of two graphs (composed of \qms) that, when used in tandem, limit the number of bad \qms that can be formed.

Finally, in Section~\ref{sec:pow}, we describe how PoW is used to provide the guarantees on IDs discussed in Subsection~\ref{sec:model}. The main challenge is defending against an adversary that wishes to store a large number of IDs for use in a massive future attack (i.e., a pre-computation attack).

Our main result is the following: 

\begin{theorem}\label{thm:main}
Assume an input graph $\G$ that satisfies $P1$ - $P4$, and that the adversary has at most a $\beta$-fraction of the computational power, for some sufficiently small but positive constant $\beta$. Then our construction using $|\Q| = O(\log\log n)$ provides the following guarantees w.h.p. over a polynomial number of join and departure events.\vspace{1pt}
\begin{itemize}[leftmargin=12pt]
\item All but an $O(1/\textup{poly}{(\log n)})$-fraction of \qms are good.\vspace{2pt}
\item All but an $O(1/\textup{poly}{(\log n)})$-fraction of IDs can successfully search for all but an $O(1/\textup{poly}{(\log n)})$-fraction of the resources.\vspace{1pt}
\end{itemize}
\noindent That is, we achieve $O(1/\textup{poly}{(\log n)})$-robustness.  
\end{theorem}

\noindent To illustrate our cost improvements, we also establish:

\begin{corollary}\label{cor:main}
Using any of the constructions in~\cite{fraigniaud:d2b},~\cite{malkhi_naor_ratajczak:viceroy}, or~\cite{naor:novel} as an input graph $\G$, our result gives the following bounds on cost.\vspace{1pt}
\begin{itemize}[leftmargin=12pt]
\item Group communication incurs $O(\textup{poly}{(\log\log n)})$ messages.\vspace{2pt}
\item Secure routing incurs $O(D\, \textup{poly}{(\log\log n)} )$ messages.  \vspace{2pt}
\item State maintenance is expected $O(\textup{poly}{(\log\log n)})$ per ID. \vspace{1pt}
\end{itemize}
\end{corollary}


\noindent Note that these are substantial improvements over the costs described in Section~\ref{sec:intro}.\smallskip\smallskip

\noindent{\bf Can we do better?} We offer some intuition for why significantly improving on our result seems unlikely.  With $|\Q|=\Theta(\log\log n)$, the probability of a bad \qm~is roughly $1/\textup{poly}(\log n)$. All constructions with $o(\log n)$ degree requires $D=\omega(\log n/\log\log n)$ IDs to be traversed in a search. Thus, the probability of encountering a bad \qm~along the  path of search is (roughly) at most $\sum_{1}^{D} 1/\textup{poly}(\log n)$ by a union bound, and this can be {\it less than $1$}.

Now, consider a smaller \qm size of $o(\log\log n / \log \log \log n)$. Then, the probability of a bad \qm~is $\omega(\log \log n/\log n)$, and over $D$ IDs traversed, a union bound no longer  bounds the probability of a failed search by less than $1$. 


In this sense, our choice of $|\Q|$ appears to be pushing the limits of what is possible when using \qms to design attack-resistant systems.\vspace{4pt}


\section{The Static Case}\label{sec:static}\vspace{2pt}
 

We first prove a result for searches without dealing with ID joins and departures. \vspace{-0pt}  

\subsection{The \Qm~Graph}\label{sec:groupgraph} \vspace{2pt}


Given an input graph $\G$, our $\epsilon$-robust construction is a \defn{\qm~graph} $\QG$ where each ID $w$ in $\G$ corresponds to \qm~$\Q_w$.  


We refer to each group in $\QG$ as \defn{blue} or \defn{red}. A blue group corresponds to a good \qm~with its neighbors correctly established, while a red group corresponds to a bad \qm~{\it\underline{or}} a \qm~that has an incorrect neighbor set.    When addressing the group graph, we use the notation $G_w$ to denote a vertex in $\QG$; however, $G_w$ may also refer to the group with leader $w$, and this will be made clear from the context.


In $\QG$, each blue \qm $G_w$ has a \defn{neighbor set}, {\boldmath{$L_w$}}, where for each $u$ that is a neighbor of $w$ in $H$, the group $G_u$ exists in $L_w$; each red group has an arbitrary neighbor set determined by the adversary.  A  search in $\QG$ proceeds over edges in $\QG$ as it would in $\G$ with the corresponding \qm members participating in the search; this is illustrated in Figure~\ref{fig:input}. For an edge $(\Q_w, \Q_v)$ between two blue groups in $\QG$, there are all-to-all links between (at least) the good members in $\Q_w$ and $\Q_v$. 
 
The following properties of $\QG$ are useful for our analysis: \vspace{3pt}
\begin{itemize}[leftmargin=10pt]
\item{\it S1.} For each ID $v$ in $H$, there is a group $G_v$ in $\QG$ where the leader $v$ of $\QG$ has the same ID in both graphs.\vspace{4pt}
\item{\it S2.} Each  group in $\QG$ is red independently with probability {\boldmath{$\pbad \leq 1/\log^k n$}} for a tunably large constant $k>0$ depending only on $d_1$; and blue otherwise. \vspace{4pt}
\item{\it S3.} Edges incident to blue groups are set according to the topology of $H$. All other edges are set by the adversary. \vspace{4pt}
\end{itemize}


The value of $\pbad$ in S2 corresponds to the probability that a \qm is bad or does not have the correct neighbor set. To provide intuition for our bound on $\pbad$, note that if we select $\Theta(\log\log n)$ IDs 
u.a.r., then the probability that a majority are bad is $O(1/\textup{poly}(\log n))$ by a Chernoff bound. A similar bound can be derived on the probability of incorrectly setting up neighbors; there is a subtlety with respect to bounding this probability, and this is discussed in our online version~\cite{JaiyeolaPSYZ17}, Appendix~\ref{app:initialization}.  Keeping $\pbad$ upper bounded by $1/\log^k n$ when IDs can join and depart is non-trivial.  We show how to do this in Section~\ref{sec:dynamic_case}.

Since the adversary controls all red groups, it is free to insert or delete edges between red groups (cf. $S3$).  However, edges involving at least one blue group are not modified.  This is because the adversary cannot modify a  good \qm's knowledge of who its neighbors are, since that knowledge is kept consistent by the good majority, although the neighboring bad \qm~may certainly ignore or corrupt incoming messages from that good \qm. \smallskip

\begin{figure}[t]\vspace{-0.6cm}
\captionsetup[subfigure]{labelformat=empty}
\hspace{0.5cm}\begin{subfigure}{0.6\textwidth} 
\vspace{0.7cm}\includegraphics[width=0.7\textwidth]{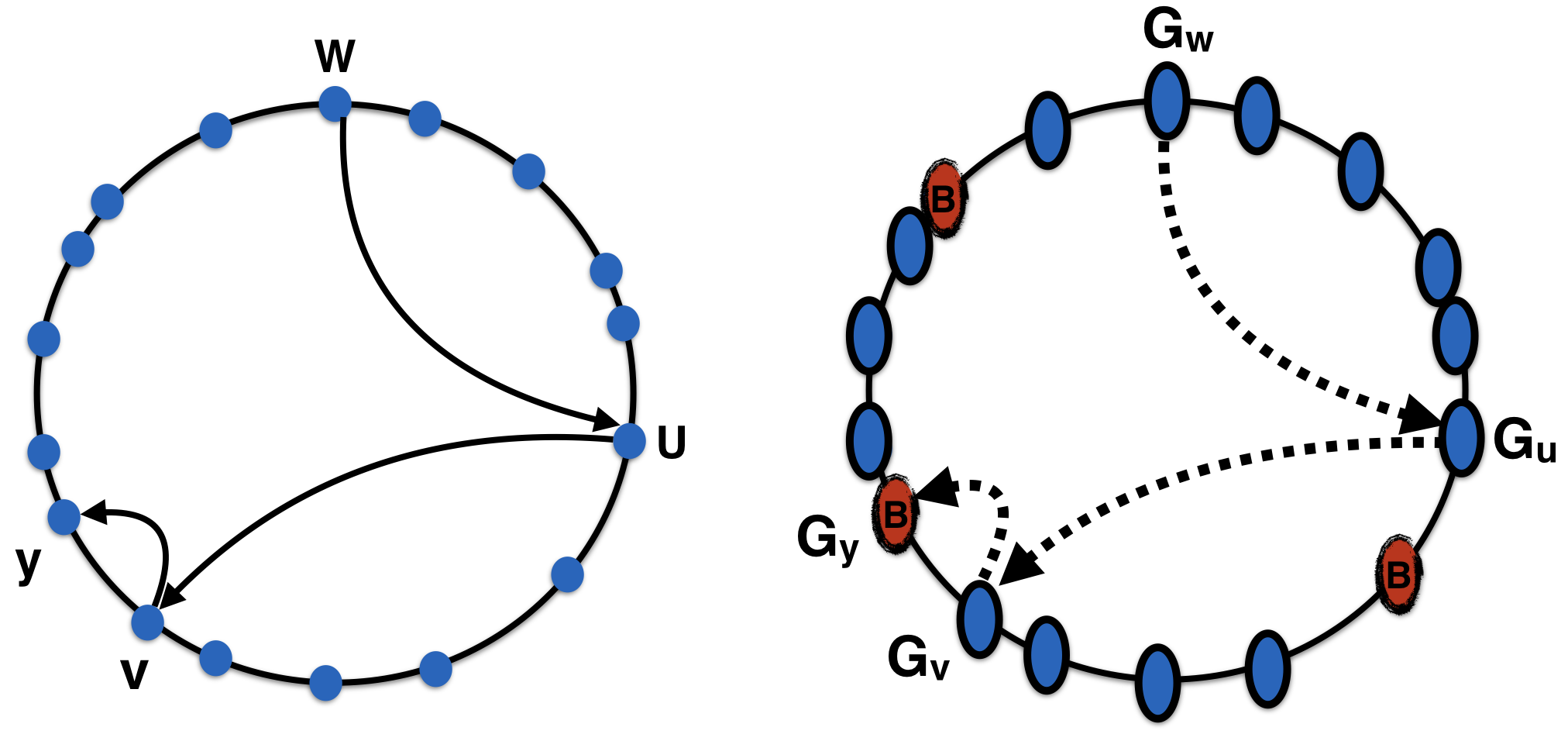} 
\end{subfigure}\hspace{0cm}
\vspace{-2pt}\caption{{\small Left: An input graph $\G$ with IDs $w, u, v,$ and $y$. Illustrated is a search initiated at $w$ and terminating at $y$; this search also traverses $u$ and $v$. Right: A \qm~graph with corresponding groups $\Q_w, \Q_u, \Q_v,$ and $\Q_y$. Red groups are marked with a ``B''. Large dashed arrows represent all-to-all links between (at least) good members of the corresponding groups.}} \label{fig:input}
\vspace{-10pt}
\end{figure}





\noindent{\bf Overview of Analysis.}   A search in group graph $\QG=(V,E)$ is said  to \defn{fail} if it traverses any red group.  Otherwise, the search will \defn{succeed}.


In $\QG$, consider the path of a search that begins at the initial group and halts either upon succeeding or encountering the first red group (in which case the search fails); we call this a \defn{\spath}.  
 
For any group $\Q_v$, we define the $\defn{responsibility}$ of $\Q_v$ to be the probability that a \spath in $\QG$ from a random  group to a random point in $[0,1)$ traverses $\Q_v$.   We denote the responsibility of group $\Q_v$ by {\boldmath{$\rho(G_v)$}}. 

Why is responsibility defined in terms of \spath{}s? The issue is that responsibility is not well-defined after a search encounters the first red group since the adversary may redirect a search to any red group after this point. For example, the adversary may have the same red group traversed by multiple different searches, thus arbitrarily inflating the number of searches that traverse this red group. This motivates  the notion of a \spath.  



\renewcommand{\baselinestretch}{1.0}
\subsection{Analysis}

\begin{lemma} \label{lem:probNode}
The following holds, \whp, for any  group $\Q_v$: $\rho(G_v) = O(\log^{c}n /n)$. 
\end{lemma}
\begin{proof}
By property P4, for any vertex $v$ in $\G$, w.h.p. any search initiated at a random vertex for a random point traverses $v$ with probability $C = O(\log^{c}n/n)$. The corresponding \spath in $\QG$ terminates either when it is successful or when the first red group is encountered; therefore, the \spath is always a subpath of the corresponding path of a search in $\G$. Also, note that any extra edges added between red groups in the $\QG$ (due to $S3$) do not affect how the \spath proceeds given that the \spath terminates at the first red group. Consequently, w.h.p., the corresponding group $\Q_v$ is traversed with probability $O(\log^{c}n/n)$.
\end{proof}




\noindent Let {\boldmath{$X$}} be a random variable that is the probability that a search that begins at a randomly-chosen \qm for a random point in $[0,1)$ fails.  The randomness of $X$ depends on which \qms are red.

\begin{lemma}\label{lem:expected}
With high probability, $E(X) = O(\pbad  \log^c n)$. 
\end{lemma}
\begin{proof}
For any group $\Q_v$, let the random variable $X_v = \rho(G_v)$ if $\Q_v$ is red, and $0$ otherwise.  By the definition of $X$, it is the case that $X \leq \sum_v X_v$.  By linearity of expectation, $E(\sum_v X_v) = \sum_v E(X_v)$.  Then, by Lemma~\ref{lem:probNode}, and the fact that each group is red with probability at most $\pbad$,  $\sum_v E(X_v) = O(\pbad  \log^c n)$.
\end{proof}



\begin{lemma} \label{l:smallX} 
With high probability, $X = O(\pbad \log^c n)$. 
\end{lemma}
\begin{proof}
For any group $\Q_v$, let $X_v=\rho(G_v)$ if $\Q_v$ is red, and $0$ otherwise.  We will bound $\sum_v X_v$, which we again note is always at least as large as $X$.  Let $f(X_1, ..., X_n) = \sum_v X_v$, where we index the groups by $1$ to $n$.  By Lemma~\ref{lem:probNode}, for any $X_i$,  $|f(..., X_i = x, ...) - f(..., X_i = x', ...)| =$ $O(\log^c n/n)$.  Thus, we can apply Theorem~\ref{thm:MOBD}  with $c^2_i = O(\log^{2c} n / n^2)$ for all $1 \leq i \leq n$. We have that:\vspace{-5pt}

$$Pr(|X - E(X)| \geq \lambda) \leq e^{-2\lambda^2 n/( d \log^{2c} n)}$$

\noindent for some constant $d>0$. Let $\epsilon>0$ be an arbitrarily small constant, and set $\lambda = \epsilon \pbad \log^c n$, which is at least $\epsilon /\log^{k-c}n$, by S2.  Then, we have: \vspace{-5pt} 

$$Pr( |X - E(X)| \geq \lambda) \leq e^{-2\epsilon^2  n / (\log^{2k} n)}$$


\noindent By Lemma~\ref{lem:expected}, \whp, $E(X) \leq d' \pbad  \log^c n$ for some constant $d'>0$, which gives the result.  \end{proof}


\begin{lemma}\label{lem:static-aer}
With high probability, any search from a random group to a random point in $[0,1)$ succeeds with probability $1- O(1/\log^{k-c} n)$. 
\end{lemma}
\begin{proof}
By Lemma~\ref{l:smallX}, for $n$ sufficiently large, w.h.p. $X = O(\pbad \log^c n) = O(1/\log^{k-c} n)$, given that $\pbad \leq 1/\log^k n$ by S2. 
\end{proof}

\section{The Dynamic Case}\label{sec:dynamic}

We now consider the case where IDs can join and depart. We still make the assumptions about IDs described in Section~\ref{sec:model}. Time is divided into disjoint consecutive windows of {\boldmath{$T$}} \defn{steps} called \defn{epochs} indexed by $j\geq 1$. 

\smallskip

\noindent{\bf Model of Joins and Departures.} We assume that $n$ IDs are always present even under churn; that is, when an ID departs, another is assumed to join; this is a popular model considered in much of the previous literature on tolerating $\Theta(n)$ Byzantine faults; for  example,~\cite{datar:butterflies,HK,naor_wieder:a_simple,fiat:making,fiat_saia:censorship,awerbuch_scheideler:group,awerbuch:towards,awerbuch:random,awerbuch:towards2}. Additionally, our results hold when the system size is $\Theta(n)$  -- that is, the size changes by a constant factor -- but we omit these details in this extended abstract.

Recall that a good group $\Q$ contains at most a $(1+\delta)\beta$-fraction of bad IDs where $\beta>0$ is a sufficiently small constant, and $\delta>0$  is a tunably small constant depending only on sufficiently large $n$. We assume the following: in any epoch, at most an $(\epsilon'/2)$-fraction of good IDs depart any group, where $\epsilon' = 1 - 2(1+\delta)\beta$. This value of $\epsilon'$ ensures that a good group retains a good majority over its lifetime.  For ease of exposition in our analysis of the dynamic case, we revise our definition of a good group to be: a group $\Q$ that begins with size $d_1\ln\ln n \leq |\Q| \leq d_2\ln\ln n$ and with at most a $(1+\delta)\beta$-fraction of bad IDs, and retains a good majority. In practice, the length of an epoch, $T$,  may be set appropriately by the system designers based on the expected rate of departures, and the value of $\epsilon'$ can be increased by increasing $d_1$.\smallskip

\noindent{\bf Algorithmic Overview.} In any epoch $j$, there are:
\begin{itemize}[leftmargin=12pt]
\item two \defn{old} \qm~graphs $\QG^{j-1}_1$ and $\QG^{j-1}_2$, each with $n$ IDs. 
\item two \defn{new} \qm~graphs $\QG^{j}_1$ and $\QG^{j}_2$, each with $\leq n$ IDs.
\end{itemize}

We emphasize that the use of two \qm~graphs per epoch is critical. A naive approach is to use a {\it single} \qm~graph in the current epoch in order to build a new \qm~graph in the next epoch.  However, this approach will fail because errors from bad \qms will accumulate over time.  Below we give some intuition for why.

Informally,  in epoch $j$, we have a process where (1) bad \qms build new bad \qms, and (2)  good \qms build bad \qms with some failure probability $p^j_f>0$ that depends on the current number of bad \qms. Therefore, in the next epoch $j+1$, the population of bad \qms has increased and so has $p_f^{j+1}$. Left unchecked, this increasing error probability will surpass the desired value of $1/\log^k n$. By using two \qm~graphs, we can upper bound $p^j_f$ by this value. 



The new \qm~graphs are built using the old \qm~graphs over the $n$ deletions and additions that occur in the current epoch $j$; we describe this in Subsection~\ref{sec:join-dep}.  By the end of epoch $j$, the old \qm~graphs $\QG^{j-1}_1$ and $\QG^{j-1}_2$ are no longer needed, and the new ones $\QG^{j}_1$ and $\QG^{j}_2$ are complete.


\subsection{Building New \Qm~Graphs}\label{sec:join-dep}

We describe how the new \qm~graphs $\QG^{j}_1$ and $\QG^{j}_2$ are created. We assume the correctness of two initial group graphs $\QG^{0}_1$ and $\QG^{0}_2$ with neighbor sets for blue groups correctly established.  This aligns with prior literature in the area and more discussion is given in our online version~\cite{JaiyeolaPSYZ17} in Appendix~\ref{app:initialization}. Later, in Section~\ref{sec:dynamic_case}, we prove that w.h.p. this construction preserves $\varepsilon$-robustness.\smallskip\smallskip

\noindent{\bf Preliminaries.} Assume the system is in epoch $j\geq 1$. Each good ID $v$ already in the system uses the same ID in both $\QG^{j-1}_1$ and $\QG^{j-1}_2$.

Any ID (already participating in the system or a newcomer) that wishes to participate in the next epoch $j+1$ must begin generating an ID  by the halfway point of the current epoch $j$. Generating this ID requires an expenditure of computational power as described in Section~\ref{sec:pow}. 

Recall from Section~\ref{sec:model} that IDs expire after a tunable period of time. Upon creation, the new ID will be \defn{active} throughout epoch $j+1$ allowing $v$ to initiate searches via $\Q_v$ and for $v$ to be added to other \qms. When $v$'s ID expires,  the \qm~$\Q_v$ (this includes $v$) should remain in both old  graphs for an additional $T$ steps. During these steps, $\Q_v$ will forward communications, but $v$ cannot initiate searches using $G_v$, nor can $v$ be added to new groups; we say that $v$'s ID is \defn{passive}.

For any group $\Q_v$, if the leader $v$ departs the system, $\Q_v$ remains. That is, the members of $\Q_v$ still persist as a the group $\Q_v$ in their respective active or passive states. We discuss departures further below in the context of updating links.

IDs are assumed to know when the system came online (i.e. step 0).\footnote{This is a fixed parameter included as part of the application, along with $T$, the hash functions, and various constants.} Since $T$ is set when the system is designed, any ID that wishes to join knows when the current epoch ends and the next one begins. Some synchronization between devices is implicit. In practice, this is rarely a problem given the near-ubiquitous Internet access (see the Network Time Protocol\cite{103043}) available to users. 

A new ID joins the new \qm~graph by a \defn{bootstrapping \qm}~denoted by {\boldmath{$\Q_{\mbox{\tiny boot}}$}}.\footnote{As with much of the literature, we do not address concurrency. Since a join or departure  requires updating only $\textup{poly}(\log n)$ links in a \qm~graph,  we assume that there is sufficient time between events to do so.}  Throughout, we assume that a joining ID knows a good bootstrapping \qm; we discuss this further in our online version~\cite{JaiyeolaPSYZ17} in Appendix~\ref{sec:boot}. \smallskip\smallskip

\noindent{\bf Making a \Qm-Membership Request.}  In epoch $j$, \qm~graphs $\QG^{j}_1$ and $\QG^{j}_2$ are built using searches in  $\QG^{j-1}_1$ and $\QG^{j-1}_2$.  An ID $w$ uses the same ID in  $\QG^{j}_1$ and $\QG^{j}_2$.


$G_w$ is added to $\QG^{j}_1$ as follows.  The $i^{th}$ member of $\Q_w$ is \suc{($h_1(w, i)$)} (recall the notation in Section~\ref{sec:model}) in the old \qm~graphs for $i = 1, ..., d_2\ln\ln n$ where $h_1$ is a secure hash function and $d_2$ is defined with respect to \qm~size in Section~\ref{sec:model}.  That is, in {\it{\underline{both}}} $\QG^{j-1}_1$ and $\QG^{j-1}_2$, a search for each successor of $h_1(w, i)$ is performed; this is executed by the bootstrapping \qm~and \suc($h_1(w, i)$) is solicited for membership in $\Q_w$. Note that if different IDs are returned by the two searches, the successor to $h_1(w, i)$  is selected.

How is $\ln\ln n$ estimated? A standard technique for estimating $\ln n$ to within a constant factor is as follows. For u.a.r. IDs, the distance $d(u,v)$ between any two IDs $u$ and $v$ satisfies $\frac{\alpha''}{n^2} \leq d(u,v) \leq \frac{\alpha' \ln n}{n}$ w.h.p., depending only on sufficiently large positive constants  $\alpha',\alpha''$ Therefore, w.h.p. $\ln\ln(\frac{1}{d(u,v)}) = \ln\ln(n) + O(1) $; this approach works even when an adversary decides to omit some (or all) of its IDs (see Chapter 4 in~\cite{young:making}), which is considered in Section~\ref{sec:dynamic_case}.


During epoch $j$, all IDs in $\QG^{j-1}_1$ and $\QG^{j-1}_2$ are active --- and will remain in a passive state over the next epoch $j+1$ --- and so can be used as members for new \qms in $\QG^{j}_1$. 

Finally, a similar process occurs to build $\Q_w$ in $\QG^{j}_2$, except that a different secure hash function, $h_2$, is used. That is, a search for the successors of $h_2(w, i)$ occurs in both  $\QG^{j-1}_1$ and $\QG^{j-1}_2$. Note that the membership of $\Q_w$  is likely different in each \qm~graph.\smallskip\smallskip



\noindent{\bf Making a Neighbor Request.} If $w$ and $u$ are neighbors in the input graph, then $\Q_w$ and $\Q_u$ should be neighbors in the \qm~graph; recall that this entails all-to-all links between the good members of both \qms. 
 To set up the neighbors of $\Q_w$, $\Q_{\mbox{\tiny boot}}$ performs a search on behalf of $w$ to locate each such neighbor $u$ in both old \qm~graphs (again, favoring the successor if the results differ). In this way, $\Q_{\mbox{\tiny boot}}$ allows $u$ (and $\Q_u$) to learn about $w$ and agree to set up a link in the respective \qm~graph.\smallskip\smallskip

\noindent{\bf Verifying Requests.} The adversary may attempt to have many good IDs join as neighbors or members of a bad \qm. This attack is problematic since good IDs will have resources consumed by maintaining too many neighbors or  joining too many \qms; this increases the state cost (see Section~\ref{sec:intro}). To prevent this attack, any such request must be verified:\smallskip

\noindent{\it Verifying a \Qm-Membership Request.}  When ID $u$ in $\QG^{j-1}_1$ is asked to become a member of \qm~$\Q_w$ in $\QG^{j}_1$, ID $u$ must verify that this request aligns with the linking rules; recall that this is assumed possible by property P3 of the input graph.



To do this verification,  $u$  performs a search on $h_1(w, i)$ in both group graphs $\QG^{j-1}_1$ and $\QG^{j-1}_2$. If either returns $u$, then the request is considered verified and $u$ becomes a neighbor of $w$; otherwise, the request is rejected.  Note that $u$ may erroneously reject a membership request; the impact on establishing good groups is addressed in Lemma~\ref{lem:good}. Conversely, $u$ may erroneously accept a membership request,  the impact of this on expected state cost is addressed in Lemma~\ref{lem:link-cost}.
\smallskip

\noindent{\it Verifying a Neighbor Request.} An ID $u$ that is asked to become a neighbor of ID $w$, and thus establish links between the members of $\Q_u$ and $\Q_w$, must also verify this request. Similar to a group-membership request,  $u$ will determine via a search in $\QG^{j-1}_1$ and $\QG^{j-1}_2$ whether $u$ should indeed be a neighbor of $w$.  If either search returns $u$, then the request is verified and $u$ becomes a neighbor of $w$; otherwise, the request is rejected. 

Note that $u$ may erroneously reject a neighbor request; this is addressed in Lemma~\ref{lem:confused}. Also, $u$ may erroneously accept a request; the impact on expected state cost is addressed in Lemma~\ref{lem:link-cost}. \smallskip

\noindent{\bf Performing a Search.} Throughout epoch $j$, each new ID $w$ performs searches only in the old \qm~graphs $\QG^{j-1}_1$ and $\QG^{j-1}_2$. Searches are performed by forwarding the request to $\Q_{\mbox{\tiny boot}}$, and executing the search from that position. Since $\Q_{\mbox{\tiny boot}}$ was active when $w$ joined, then $\Q_{\mbox{\tiny boot}}$ should remain -- even if in a passive state -- to facilitate searches for another $T$ steps. 

Over the duration of epoch $j$,  a \qm $\Q_w$ may not be able to reliably perform searches in the new \qm~graphs $\QG^{j}_1$ and  $\QG^{j}_2$ since they are still under construction. For example, $w$ might be the first ID to join  $\QG^{j}_1$ and  $\QG^{j}_2$. Once epoch $j+1$ starts, the new \qm~graphs $\QG^{j}_1$ and $\QG^{j}_2$ are to be used. At this point, group $\Q_w$ will initiate any search using its own links in these graphs, rather than relying on $\Q_{\mbox{\tiny boot}}$ which may no longer be present in the system.\smallskip

\noindent{\bf Updating Links.}  When a new ID (and its corresponding group) is added to the group graphs $\QG^{j}_1$ and  $\QG^{j}_2$ that are under construction, a \qm~$\Q_w$ must update its neighbor links in $\QG^{j}_1$ and  $\QG^{j}_2$ if this new ID is a better match as a neighbor under the linking rules. This update is done via searches by a bootstrapping group in the old \qm~graphs $\QG^{j-1}_1$ and  $\QG^{j-1}_2$.  

Conversely, if $\Q_w$ links to some group $\Q_v$ whose members all depart --- note that $\Q_v$ must consist entirely of bad IDs given our model of churn --- then $\Q_w$ treats that link as null (until perhaps a join operation requires an update).\footnote{In practice, neighbors will periodically ping each other in order to check that neighbors are still alive.}

Groups in the old group graphs $\QG^{j-1}_1$ and  $\QG^{j-1}_2$ do not update links. For an old group graph, if all members of $\Q_v$ depart, $\Q_w$  treats that link as null until the group graph expires.



\subsection{Analysis}\label{sec:dynamic_case}

In this section, we analyze the construction of new \qm~graphs. Due to space constraints, some of our proofs are provided in our online version~\cite{JaiyeolaPSYZ17} in Appendix~\ref{app:proofs}.


Properties P1-P4 of input graph $\G$ play an important role in the design of the corresponding group graph. However, a prerequisite to these properties is that {\it all} IDs are selected uniformly at random (see Section~\ref{sec:model}), which is untrue if the adversary chooses to add only some of its bad IDs; for example, maybe only bad IDs in $[0, \frac{1}{2})$ are added by the adversary. Intuitively, this should not interfere with any of the properties; we now formalize this intuition.




In the following, we consider $\G'$ to be a modified input graph which uses the same construction as $\G$, but is subject to an adversary that only includes a subset of its IDs, from a larger set of u.a.r. IDs. We note that $\G'$ is not necessarily a subgraph of $\G$; the omission of bad IDs can result in a different topology.

\begin{lemma}\label{lem:congestion_subtle}
Consider a graph $\G'$, where the IDs are formed from two sets:
\begin{itemize}[leftmargin=12pt]
\item $\mathcal{N}_1$ consists of at least $(1-\beta)n$ IDs selected u.a.r. from $[0,1)$. 
\item  $\mathcal{N}_2$ is an arbitrary subset of at most $\beta n$ IDs selected u.a.r. from $[0,1)$.  
\end{itemize}
W.h.p., under the same construction as the input graph $\G$, graph $\G'$ has properties P1 - P4.
\end{lemma}

Throughout,  the above result is assumed --- that properties P1-P4 continue to hold if the adversary includes only a subset of its IDs --- even if we do not always make it explicit.  For example, P1 is important throughout our arguments/construction, P2 is used in Lemma~\ref{lem:random},  P4 in Lemma~\ref{lem:dyn-aer}, and P3 in Lemma~\ref{lem:link-cost}.\smallskip

As described above, for an ID $u$, there are searches on random key values, via hashing under the random oracle assumption, in order to find members for \qm~$\Q_u$. But if a key value maps to a bad ID, then this results in a bad member added to the \qm. We can  bound the probability of this event as follows. 

\begin{lemma}\label{lem:random}
W.h.p.  a random key value in an old \qm~graph maps to a bad ID with probability at most $(1+\delta'')\beta$ for an arbitrarily small constant $\delta''>0$ depending only on sufficiently large $n$.
\end{lemma}

\noindent Let {\boldmath{$\psearch = O(1/\log^{k-c}n)$}} be the probability that a search for a random key in an old \qm~graph  fails; recall Lemma~\ref{lem:static-aer}.  

\begin{lemma}\label{lem:good}
Each group in a new group graph is bad with probability at most $O(\psearch^2d_2\log\log n + 1/\log^{d'}n)$  for  a tunable constant $d'>0$ depending on $d_2$. 
\end{lemma}
\begin{proof}
For a new ID $w$, there are two ways in which a search for a member of $\Q_w$ may result in a bad ID. 
First, a search for a \qm~member may fail; that is, the search encounters a bad \qm. Given a point $h_1(w, i)$, the probability that both searches in the old \qm~graphs fail is at most $\psearch^2$. By a union bound, the probability of such a dual failure occurring over $d_2\ln\ln n$ searches is  $O(\psearch^2\,d_2\log\log n)$.

Second, the search succeeds but returns $\suc(h(w,i))$ where  $\suc(h(w,i))$ is a bad ID, 
even though its \qm is good.  By the random oracle assumption, $h(w,i)$ is a random point.  Thus, this event occurs with probability at most $(1+\delta'')\beta$ by Lemma~\ref{lem:random}, for an arbitrarily small constant $\delta''>0$ given sufficiently large $n$. Over $d_2\ln\ln n$ searches, the expected number of such events is at most $(1+\delta')\beta d_2\ln\ln n$. The probability of exceeding this expectation  by more than a small constant factor (and adding too many bad IDs)  is $O(1/\log^{d'}n)$ by a  Chernoff bound, where the constant $d'>0$ is tunable depending only on sufficiently large $d_2$.

Finally, the ID being asked to join may reject the request. This occurs if both searches used to verify the request fail. By a union bound, this happens with probability at most $O(\psearch^2\,d_2\log\log n)$.
\end{proof}


A blue \qm~$\Q_w$ should link to all \qms  in the neighbor set $L_w$.  Recall that $|L_w| = O(\log^{\gamma} n)$ for some constant $\gamma> 0$.  If $\Q_w$ (1) links to any \qm  not in $L_w$, or (2) fails to link to any \qm  in $L_w$, then $\Q_w$ is said to be \defn{confused}. We can bound the probability of a confused \qm; the proof is provided in our online version~\cite{JaiyeolaPSYZ17} in Appendix~\ref{app:proofs}.

\begin{lemma}\label{lem:confused}
Each \qm~in a new \qm~graph is confused independently with probability at most $O(\psearch^2 \log^{\gamma} n)$.
\end{lemma}



\noindent We now prove that w.h.p. each new \qm~graph is $\varepsilon$-robust. 


\begin{lemma}\label{lem:dyn-aer}
Assume the old group graphs are $\varepsilon$-robust and that the adversary has at most $\beta n$ u.a.r. IDs. Then, for $k\geq 2c+\gamma$, w.h.p., each new \qm~graph is $\varepsilon$-robust.
\end{lemma}
\begin{proof}
In our analysis, a \qm that is bad or confused is a red group;  otherwise, the group is good and not confused, and this a blue group.  To show $\varepsilon$-robustness of the new group graph, we prove that the probability of creating a red \qm in the new group graph is at most $\pbad \leq 1/\log^k n$ for a tunable constant $k>0$. By Lemmas~\ref{lem:good} and~\ref{lem:confused}, each \qm~is red independently with probability at most:
{\small
\begin{eqnarray*} 
&\leq&O\left(\psearch^2 \log^{\gamma} n\right) + O\left(\psearch^2d_2\log\log n +  1/\log^{d'}n \right) \\
& \leq & O\left(\frac{\log\log n}{\log^{2(k-c)-\gamma}n} + \frac{1}{\log^{d'}n} \right) \leq \frac{1}{\log^k n}
\end{eqnarray*}}
\noindent The last line follows by setting $d_2$ to be sufficiently large such that $d'$ exceeds $k$.  Note  that $d_2$ is fixed at the beginning and never needs to be changed throughout the lifetime of the network. Then, setting $k>2c+\gamma$ to be a sufficiently large constant yields the necessary value of $\pbad$ to establish the inequality.

This implies that all but an $o(1)$-fraction of groups are good and not confused. Furthermore, once we have this bound on $\pbad$, the remaining proof is equivalent to that of  Lemma~\ref{lem:static-aer}.
\end{proof}



\noindent The next lemma  bounds the amount of state a good ID maintains due to (1) membership in \qms, and (2) being a neighbor of a \qm. This is done by analyzing the verification process described in Section~\ref{sec:join-dep}; see~\cite{JaiyeolaPSYZ17} for the proof.

\begin{lemma}\label{lem:link-cost}
In expectation, each good ID $w$ in a \qm~graph is a member of $O(\log\log n)$ \qms and maintains state on  $O(|L_w|)$ \qms that are either neighbors or have $w$ as a neighbor.
\end{lemma}


\noindent We can now prove Theorem~\ref{thm:main}. \vspace{-5pt}
\begin{proof}
Lemma~\ref{lem:dyn-aer} guarantees w.h.p. that in the new \qm~graphs, all but a $1/\textup{poly}{(\log n)}$-fraction of \qms~are good, and all but a $1/\textup{poly}{(\log n)}$-fraction of IDs can search for all but a $1/\textup{poly}{(\log n)}$-fraction of the resources. 

Given that \qms have size $O(\log\log n)$, it follows that \qm~communication incurs $O( (\log\log n)^2)$ messages. Recall that secure routing proceeds via all-to-all communication between members of \qms and that searches have maximum length $D$ (P1 in Section~\ref{sec:model}). Thus, the message complexity is $O(D(\log\log n)^2 )$.

To bound the expected state cost, we invoke Lemma~\ref{lem:link-cost}.  Each good ID $w$ belongs to $O(\log\log n)$ \qms~in expectation. This implies $O((\log\log n)^2)$ expected state cost to keep track of the members of these groups.

In terms of links to and from other \qms, $w$ maintains state on $O(|L_w|)$ \qms in expectation.  The constructions for $\G$ defined in~\cite{naor:novel},~\cite{fraigniaud:d2b}, or~\cite{malkhi_naor_ratajczak:viceroy} provide the properties P1-P4, but with a bound of $O(1)$ expected degree. Using any of these constructions, the state cost incurred by these neighboring groups  is $O(\log\log n)$ in expectation. Thus, the total expected state cost is $O((\log\log n)^2) + O(\log\log n) = O((\log\log n)^2)$. Corollary~\ref{cor:main} follows immediately.\vspace{-0pt}
\end{proof}



\section{Computational Puzzles}\label{sec:pow}\vspace{-0pt}

Up to this point, we have assumed that the adversary can inject into each new \qm~graph at most $\beta n$ bad IDs with u.a.r. values, and that these IDs can be verified and forced to expire after a period of time; recall our discussion in Section~\ref{sec:model}. We now remove these assumptions. Given space constraints, we limit our discussion to the main ideas of how to use computational puzzles to guarantee these properties.\vspace{-2pt} 





\subsection{Generating an ID}\label{subsec:generatingID}\vspace{-3pt} 

All participants are assumed to know two secure hash functions, $f$ and $g$, with range and domain $[0,1)$ and that both hash functions satisfy the random oracle assumption. 

In the current epoch $i$, ID $w$ is assumed to possess a ``globally-known'' random string $r_{i-1}$ of $\ell \ln n$ bits. By ``globally-known'', we mean known to all good IDs except the $1/\textup{poly}(\log n)$-fraction from our earlier analysis. We motivate $r_{i-1}$ and describe how it is generated in Subsection~\ref{subsec:randomstring}. 

Starting at step $T/2$ in the current epoch, each good ID begins generating a new ID for use in the next epoch,  as described below. \smallskip

\noindent{\bf Description of ID Generation.} To generate an ID, a good ID $w$ selects a value $\sigma_w$ of $\ell \ln n$ random bits (matching the length of $r_{i-1}$). Then, $w$ XORs these two strings to get $\sigma_w \oplus r_{i-1}$, and checks if $g(\sigma_w \oplus r_{i-1}) \leq \tau$; if so, then $f(g(\sigma_w \oplus r_{i-1}))$ is a valid ID. We assume the value $\tau$ is set small enough such that w.h.p.  $(1\pm \epsilon)T/2$  steps are required to find a $\sigma_w$ that satisfies this inequality, where $\epsilon>0$ is a tunable (small) positive constant and $T>0$ is a parameter set when the system is initialized. 





The value of $T$ can be large to amortize the cost of forcing IDs to depart (and possibly rejoin) over a long period of time; for example, $T > n$, since new group graphs are being built over $T$ steps. Given an application domain, designers may estimate the rate of churn for their application and set a (loose) upper bound on $n$, then they can set $T$ accordingly.\smallskip

\noindent{\bf Why Use Two Hash Functions?} Consider using a single secure hash function $f$ to assign IDs; that is, if $g(x) < \tau$, then $x$ is a valid ID.   Then, for example, the adversary may restrict itself to small inputs $x$ in order to confine its solutions to yielding small IDs. In other words, the IDs obtained by the adversary will not be u.a.r. from $[0,1)$. This can be solved via composing two secure hash functions, $f$ and $g$, as described above. See~\cite{JaiyeolaPSYZ17} for the proof of the following:

\begin{lemma}\label{lem:limit-comp}
W.h.p., the adversary generates at most $(1+\epsilon)\beta n$ IDs over $(1\pm \epsilon)(T/2)$ steps and these IDs are u.a.r. in $[0,1)$.
\end{lemma}

\noindent We note that Lemma~\ref{lem:limit-comp} implies that the adversary might be able to generate up to $3(1+\epsilon)\beta n$ IDs for use in the next epoch: computing over the $T/2$ steps in the last half of the previous epoch and the $T$ steps prior to the end of the current epoch. However, we can revise the adversary's power from $\beta$ to $\beta/3$, and results in Sections~\ref{sec:static} and~\ref{sec:dynamic} hold. Note that all such IDs will be invalidated when the next random string is created.   \smallskip



\noindent{\bf ID Verification.} Upon receiving a message from some ID $w$, a good ID $u$ \defn{verifies} $w$'s ID.  This could be done naively by having $w$ send $\sigma_w$ to $u$ who checks that  $g(\sigma_w \oplus r_{i-1}) \leq \tau$ and that $f(g(\sigma_w\oplus r_{i-1}))$ evaluates correctly to the claimed ID (note that $u$ already has $r_i$ since it is globally-known). Unfortunately, this allows $u$ to steal $\sigma_w$ if $u$ is bad. 

To avoid this issue, we can use a zero-knowledge scheme for revealing the pre-image of the hashing; such a scheme is provided for the SHA family~\cite{jawurek:2013}. This allows $w$ to prove the validity of $\sigma_w$ without revealing it.

If $w$'s ID fails verification, then $u$ simply ignores $w$ going forward. Note that $w$'s current ID will not be valid in the next epoch since it is signed by the older string $r_{i-1}$ (rather than the next globally-known random string $r_i$); that is, $w$'s ID will have expired.  IDs that are not verified are effectively removed from the system; they may consort with bad IDs, but they have no interactions with good IDs.
\vspace{-2pt}


\subsection{Generating Global Random Strings}\label{subsec:randomstring}\vspace{-2pt}

Imagine if no random string was used in the creation of IDs described above in Subsection~\ref{subsec:generatingID}. The adversary would know the format of the ID-generation puzzles, and so could spend time computing a large number of IDs, and then use these IDs all at once to overwhelm the system at some future point.  This is a \defn{pre-computation attack}.  

Signing IDs  with a random string prevents such an attack as it is impossible for the adversary to know far in advance how to generate IDs. We provide a protocol where random strings are generated and propagated in the system to be used in ID generation.  Due to space constraints, this content is provided in our online version~\cite{JaiyeolaPSYZ17} in Appendix~\ref{app:grs}.  We show the following:\vspace{0pt}


\begin{lemma}\label{lem:propagate}
W.h.p., the protocol for propagating strings (i) guarantees that, for each good ID $w$, its string used for generating an ID is known to each good ID, (ii) the number of strings stored by each ID is  $O(\ln n)$, and (iii) has message complexity $\tilde{O}(n \ln T )$.\vspace{0pt}
\end{lemma}

\noindent{}Note that, averaged over a sufficiently large epoch, this message cost is low. \vspace{-0pt}

\section{Conclusion and Future Work}\vspace{-3pt}

We showed that \qms~of size $O(\log\log n)$ can be used to tolerate a powerful Byzantine adversary. Our result utilizes PoW to limit the number of IDs the adversary controls; however, this imposes a computational overhead on participants. 
An open question is whether the computational costs can be reduced. Might there be a way to avoid the continual solving of puzzles? Is there an approach that would only utilize puzzle solving when malicious IDs are present? 

Another problem that deserves attention is providing a detailed mechanism for bootstrapping in the presence of a Byzantine adversary. Such a result would likely benefit prior work as well as our own; we discuss this in our online version~\cite{JaiyeolaPSYZ17} in Appendix~\ref{sec:boot}.

Finally, is it possible to show a lower bound on group size of $\Omega(\log\log n)$? We provided intuition for why this may be the smallest group size that admits strong security guarantees, but proving this appears challenging.



{\small


}





\section*{Appendix}


\section{Background Material}~\label{sec:background}

In this section, we provide background material that may provide useful context for our results.\medskip

\noindent{\bf The Search Algorithm.} We treat the search algorithm abstractly; it takes as input a value in $[0,1)$ and returns a ID. In practice, returning a  ID implies returning the information necessary for contacting that corresponding owner of the ID (i.e., the machine) over the Internet (for example, learning the port number and IP address).

A search initiated by ID $w$ is a request that will typically require contacting other IDs on its way to the final recipient $y$, who holds the resource being sought by $w$. Typically, searches are recursive, so $w$'s request is forwarded through multiple IDs before reaching $y$. Alternatively,  a search can be iterative, so $w$ contacts a sequence of IDs each of which informs $w$ how to make partial progress towards reaching $y$. This is not terminology particular to this area, the reader will likely have been exposed to these ideas in a  networks course, perhaps in the context of name resolution queries under the Domain Name System. Both $w$ and $y$, along with any IDs that participate in the search (by forwarding or resolving part of the search) are said to be IDs that are \defn{traversed} by the search. \medskip

\noindent{\bf Resources.} To flesh out the notions of a \defn{resource} and a \defn{nearest-clockwise ID}, we consider the popular Chord overlay~\cite{stoica_etal:chord}. Each ID maintains links to $O(\log n)$ neighbors, and these neighbors are used to perform recursive searches. Assume that ID $w$ wants to search for, say, a song file; that is, the song file is the \defn{resource}. 

To find this resource, $w$ applies a globally-known hash function to the title of the song to obtain a key value $K$. Note that there are no tricks here; the hash function is globally known because it comes with the software downloaded onto a user's computer that allows it to participate in the overlay, and it is not changed over the lifetime of the system. Additionally, one could generate key values for resources in other ways, but using the title is a typical example.

For simplicity, keys are assumed to be from a normalized range $[0,1)$, and this is viewed as a \defn{unit ring} where moving from $0$ to $1$ corresponds to moving \defn{clockwise} around this ring. 

This key value $K$ is placed in the message sent from $w$ to one of its neighbors, say ID $u$. If ID $u$ holds the song file corresponding to $K$, ID $u$ can return the song file directly to $w$. But, otherwise, $u$ forwards the request onto one of its own neighbors. And so on, until a ID $y$ is located that holds the file corresponding to $K$. At each forwarding step, the selection of which neighbor to forward the request to is a function of $K$; this is dependent on the system used (Chord, Viceroy, etc.). 

ID $y$ --- in our example, the ID who holds the song file on its hard drive --- is \defn{responsible} for that resource. A simple and often-used rule is that the ID which is \defn{closest-clockwise} to $K$ is responsible for the resource; this is also referred to as the \defn{successor} of $K$ in the literature. That is, $K$ is a point in the ID space $[0,1)$ viewed as a unit ring, and if we slide clockwise from $K$ on this unit ring, then the first ID encountered is responsible for the corresponding resource. With modifications, a similar framework applies to situations where the resource is a shared printer or some other service.\medskip

\noindent{\bf IDs and Neighbors.}  An ID is a value in $[0,1)$; that is, a point somewhere on the unit ring. We describe how an ID is generated for our particular construction using PoW in Section~\ref{subsec:generatingID}. 

Note that a participant (i.e., a machine) has a physical location in the real world, while an ID is a virtual location in the overlay. Therefore, a good machine can hold more than $1$ ID. Indeed, our construction requires;  a good machine $u$ can hold an ID for two old group graphs, and an ID for two new group graphs.  

What does this mean? Consider $u$'s ID in an {\underline{old}} group graph. This ID is linked to neighboring groups as specified by the input graph topology, and also holds links to members of groups to which it belongs. Similarly, $u$'s ID in a {\underline{new}} group graph is also linked to neighboring groups  according to the appropriate topology, and to members of group to which it belongs. In other words, each ID is treated as a separate entity in the respective graph/network.

From the perspective of the machine $u$ that generates these IDs, $u$ must must commit resources to maintaining links for each ID it holds. Say $u$'s ID links to $w$'s ID, then in practice this means that $u$ maintains state on, say, a TCP connection between $u$ and $w$. Therefore, each ID held by $u$ incurs a cost in terms of state maintenance and bandwidth.

In contrast to good machines, the adversary may attempt to create many, many IDs within a single group graph. One danger of this is that the number of  bad IDs may now be larger than the number of good IDs; see mention of the Sybil attack under related work.  In this case, we cannot hope to create groups with a good majority. Given this, a common assumption is that the adversary  holds a minority of the computational power, and therefore computational puzzles allow us to mitigate this attack. The security of Bitcoin, Ethereum, and many blockchain constructions are based on this same assumption.

This is the convention we adopt too, and we prove in Section~\ref{subsec:generatingID} that the adversary holds (roughly) at most $\beta n$ IDs in the system at any given time. This follows from our model section, where we describe the adversary as possessing at most a $\beta$-fraction of the computational power. \smallskip


\section{Proofs for Section~\ref{sec:dynamic_case}}\label{app:proofs}

\noindent{\bf \textsf{Lemma}}~{\bf \textsf{\ref{lem:congestion_subtle}}}. {\it Consider a graph $\G'$ where the IDs are formed from two sets:
\begin{itemize}
\item $\mathcal{N}_1$ consists of at least $(1-\beta)n$ IDs selected u.a.r. from $[0,1)$. 
\item  $\mathcal{N}_2$ is an arbitrary subset of at most $\beta n$ IDs selected u.a.r. from $[0,1)$.  
\end{itemize}
W.h.p., under the same construction as the input graph $\G$, graph $\G'$ has properties P1 - P4.
}
\begin{proof}
View the ID space as a unit ring, and place on it the IDs from $\mathcal{N}_1$ and $\mathcal{N}_2$. Let the total number of IDs be $m$ where $m \geq (1-\beta)n$. 

Moving clockwise from any ID, consider a contiguous interval of length $(\lambda \ln m)/m$ where $\lambda>0$ is any constant. Since IDs in $\mathcal{N}_2$ are selected from a larger set of IDs u.a.r in $[0,1)$, intuitively the adversary's choice of $\mathcal{N}_2$ cannot significantly change the density/sparseness of IDs on the ring. By Chernoff bounds, regardless of how $\mathcal{N}_2$ is selected, and for any $\lambda>0$, the following holds:
\begin{itemize}
\item with probability at least $1-(m)^{-\lambda/12}$, every interval contains at least $(\lambda/2)\ln m$ IDs, and
\item with probability at least $1-m^{-\lambda/12}$,  every interval  contains at most $(3\lambda/2)\ln m$ IDs.
\end{itemize}
A set of $m$ IDs with locations on the ring that satisfy these two properties is a \defn{\boldmath{$\lambda$}-well-spread placement}. Observe that no matter how the adversary chooses $\mathcal{N}_2$, w.h.p. the adversary's influence on the distribution of IDs is characterized by some $\lambda$-well-spread placement. In other words, we can ignore the adversary since the issue now reduces to: what is the probability of a $\lambda$-well-spread placement that degrades a property in $\G'$?

We argue by contradiction as follows. Recall the guarantee that the input graph $\G$ has some property $P$ with probability at least $1- 1/m^{c'}$ for a constant $c'>0$ (Section~\ref{sec:static}). Now assume there one or more $\lambda$-well-spread placements that violate this property $P$ for $\G'$, and that these occur with aggregate probability $1/m^d > 1/m^{c'}$ for some positive constant $d$.  But this yields a contradiction since, with probability at least $1/m^d > 1/m^{c'}$,  placing $m$ IDs  u.a.r. on the ring would yield one of these placements for the input graph $\G$ and, therefore, violate the guarantee of property $P$ for $\G$.
\end{proof}

\noindent{\bf \textsf{Lemma}}~{\bf \textsf{\ref{lem:random}}}. {\it W.h.p.  a random key value in an old \qm~graph maps to a bad ID with probability at most $(1+\delta'')\beta$ for an arbitrarily small constant $\delta''>0$ depending only on sufficiently large $n$.}

\begin{proof}
By property P2 of the input graph $\G$, w.h.p. a randomly chosen ID in $\G$ is responsible for at most a $(1+\delta'')/n$-fraction of the key values for an arbitrarily small $\delta''>0$ depending on sufficiently large $n$ (and, by Lemma~\ref{lem:congestion_subtle}, this holds even if the adversary does not add all of its bad IDs). Since the IDs of the adversary are u.a.r., the $\beta n$ bad IDs are responsible for at most a $(1+ \delta'')\beta$-fraction of the key values.
\end{proof}

\noindent{\bf \textsf{Lemma}}~{\bf \textsf{\ref{lem:confused}}}. {\it Each \qm~in a new \qm~graph is confused independently with probability at most $O(\psearch^2 \log^{\gamma} n)$}
\begin{proof}
Since the bootstrapping \qm~is good, the only way in which a \qm~$\Q_w$ is confused about a member of $L_w$ is if (1) the two searches for a neighbor in the {\it old} \qm~graphs both fail or (2) the corresponding group that is asked to be a neighbor erroneously rejects the request.  

We analyze these two cases below, but first we elaborate on a subtle point. As described earlier, updates to links occur as  the new group graphs are being constructed. For instance, if $G_v$ is the first group to be added to a new group graph, then $G_v$ will need to do searches to update its neighbors correctly. Updates may occur more than once over time as the new group graph is being constructed. 

Throughout this construction, there may be times where $G_v$ does not link to the correct neighbors due to a search failing in the old group graphs. Importantly, the probability of failure in the old group graphs is bounded, and is not impacted by any confused groups in the new group graphs being built. So, it is possible for $G_v$ to be temporarily confused in the new group graph.  But, importantly,  we only care about the final selection of $G_v$'s neighbors. Thus, when the final group is added to the new group graph, $G_v$'s subsequent update will lead to this final selection, and the probability of confusion at that point is what must be bounded.

We now bound the probability of cases (1) and (2). 
Applying Lemma~\ref{lem:static-aer}, case (1) occurs with probability at most $\psearch^2$ per element of $L_w$. Over $O(\log^{\gamma} n)$ potential neighbors,  a union bound limits the probability of this occurring over all elements of $L_w$ by $O(\psearch^2 \log^{\gamma} n)$.  For case (2), the group asked to be a neighbor will perform two searches per request and only reject if both fail; therefore, we get the same bound as in case (1). 
\end{proof}

\noindent{\bf \textsf{Lemma}}~{\bf \textsf{\ref{lem:link-cost}}}. {\it In expectation, each good ID $w$ in a \qm~graph is a member of $O(\log\log n)$ \qms~and maintains state on  $O(|L_w|)$ \qms~that are either neighbors or have $w$ as a neighbor.}
\begin{proof}
We perform our analysis with respect to a good ID $w$.  First, we analyze the state cost incurred by $w$ due to its membership in various groups. Second, we analyze the state cost incurred by $w$ due to (i)  $w$'s links to its neighbors, and (ii) those IDs that link to $w$ as a neighbor.\smallskip\smallskip

\noindent{\it State Cost of \Qm-Membership.} The expected number of \qms to which $w$ belongs is at least $d''\log\log n$, for some constant $d''>0$, given that such requests are distributed among all IDs u.a.r. (via the random oracle assumption) and each \qm~requires $O(\log\log n)$ members.

For some ID $u$, consider a membership request for $G_u$ using the point $h_1(u, i)$. If this returns $w$ (as the successor of this point $h_1(u, i)$), then $w$ accepts the membership request. Given the $\varepsilon$-robustness guarantee in old \qm~graphs, with probability at least $1 - O(1/\log^{k-c} n)$, this acceptance is correct (recall Lemma~\ref{lem:static-aer}). In other words, the probability of accepting an erroneous member request is at most $1/\textup{poly}(\log n)$ where we can tune this polynomial.

How many member requests does $w$ receive?  We reiterate a well-known argument: since IDs are u.a.r., the probability that two IDs are separated by more than a $(c''\ln n)/n$ distance is at most $(1 - (c''\ln n)/n)^n \leq 1/n^{\Theta(c'')}$, for some constant $c''>0$.  Therefore, if $h_1(u, i)$ lies outside of the interval $\mathcal{I} = [w - (c''\ln n)/n, w)$, then w.h.p. $w$ is not the successor of $h_1(u, i)$ and cannot be a valid contender for membership. 

Given that $h_1(u, i)$ is u.a.r. (given the random oracle assumption) the probability that $h_1(u, i)$ maps to this interval is  $(c''\ln n)/n$. By property P3,  $i = O(\log^{\gamma} n)$ for some constant $\gamma>0$, and so over at most $n$ IDs that might own valid IDs in $\mathcal{I}$, there are $O(\log^{\gamma} n)$ such requests received by $w$; this holds w.h.p. by a standard Chernoff bound. It follows that $w$ erroneously accepts   $O(\log^{\gamma} n)/\textup{poly}(\log n)$ $= O(1)$ malicious requests in expectation so long as $k$ is sufficiently large with respect to $c$.
\smallskip\smallskip

\noindent{\it State Cost of $w$'s Neighbors.}  In each \qm~graph, $w$ links to $O(|L_w|)$  \qms~as neighbors. \smallskip\smallskip 

\noindent{\it State Cost of Other Groups Linking to $w$ as a Neighbor.}  Via Lemma~\ref{lem:congestion_subtle}, properties P1 and P3 guarantees that $w$ can determine via searches whether it should indeed be a neighbor of some ID $u$, and there are at most $\textup{poly}(\log n)$ such IDs $u$. Using the old \qm~graphs, wherein $\varepsilon$-robustness is guaranteed w.h.p.,  $w$ initiates a search to check that it should indeed be a neighbor; note that the adversary cannot lie about its IDs since they are verifiable. With a tunable probability at least $1 - O(1/\log^{k-c} n)$ ID $w$ can detect if the request is erroneous. Therefore, in expectation, the number of erroneous acceptances is at most $o(|L_w|)$ so long as our constant $k$ is sufficiently large (recall Lemma~\ref{lem:static-aer}).
\end{proof}


\section{Generating and Disseminating Global Random Strings}\label{app:grs}

\noindent{\bf \textsf{Lemma}}~{\bf \textsf{\ref{lem:limit-comp}}}. {\it 
W.h.p., the adversary generates at most $(1+\epsilon)\beta n$ IDs over $(1\pm \epsilon)(T/2)$ steps and these IDs are u.a.r. in $[0,1)$.}
\begin{proof} 
Since the adversary has $\beta n$ computational power to expend over this epoch, w.h.p. it can generate at most $(1+\epsilon)\beta n$ solutions $\sigma_v$ such that $g(\sigma_v \oplus r_{i-1}) \leq \tau$ within $(1\pm \epsilon)(T/2)$  steps where the constant $\epsilon>0$ can be made arbitrarily small depending only on sufficiently large $n$. By the random oracle assumption, w.h.p. applying $f$ to these solutions yields at most $(1+\epsilon)\beta n$ IDs u.a.r. from $[0,1)$.
\end{proof}


\noindent{\bf Generating Random Strings.}  Over epoch $i$, all good IDs generate random strings and those corresponding to the smallest output under $h$ (a secure hash function) are collected independently by each ID $w$ to create a \defn{solution set}  $R^w_i$. To generate a string in epoch $i$, a ID $w$ uses a string $r_{i-1}$ --- the globally-known string from the previous epoch  ---  and an individually generated random string, {\boldmath{$s_w$}}, to compute the \defn{output} {\boldmath{$t_w$}} $= h(s_w \oplus r_{i-1})$, where $\oplus$ indicates XOR.

\smallskip\smallskip


\noindent{\bf Bins and Counters.} To facilitate our discussion of how to propagate strings and ease our subsequent analysis, we describe a system of bins and counters maintained by each good ID $w$. The \defn{bins} {\boldmath{$B_j$}} correspond to intervals in the ID space  where  $B_j = [1/2^{j}, 1/2^{j-1} )$ for $j = 1, 2, ..., b\ln(n\,T)$ where $b\geq 1$  is a sufficiently large constant. Since $T$ is known and there are standard techniques for obtaining a constant-factor approximation to $\ln n$, calculating $\ln(n T) = \ln(n) + \ln(T)$ to within a constant factor is possible.\footnote{A standard technique for estimating $\ln n$ to within a constant factor is as follows. For u.a.r. IDs, the distance $d(u,v)$ between any two IDs $u$ and $v$ satisfies $\frac{\alpha''}{n^2} \leq d(u,v) \leq \frac{\alpha' \ln n}{n}$ w.h.p., depending only on sufficiently large positive constants  $\alpha',\alpha''$ Therefore, w.h.p. $\ln(\frac{1}{d(u,v)}) = \Theta(\ln n)$ and this holds even when an adversary decides to omit some (or all) of its IDs (see Chapter 4 in~\cite{young:making}).}

Each bin $B_j$ has an associated \defn{counter} {\boldmath{$C_j$}}. Consider that $w$ receives a string $s_u$ with corresponding output $t_u$ that falls within the interval defined by $B_j$; we say that $B_j$ \defn{contains} $t_u$. If $t_u$ is smaller than the other values $w$ has seen so far contained in $B_j$, and $C_j \leq c_0\ln n$ for some sufficiently large constant $c_0\geq 1$, then $w$ increments $C_j$ and forwards the corresponding string $s_u$ onto its neighbors. After $C_j = c_0\ln n$, no value landing within $B_j$ is ever forwarded.

The intuition is that, if $c\ln n$ strings are found with ``record-breaking'' outputs in $B_j$, then w.h.p. smaller strings exist with outputs belonging to $B_{j+1}$. In other words, those strings corresponding to $B_j$ will not be candidates for a globally-known string, and so they can be ignored.\smallskip\smallskip


\noindent{\bf Protocol for Propagating Strings.}  The propagation of strings is broken into \defn{phases} which make up the first half of an epoch. We describe the protocol for good IDs (although bad IDs can deviate arbitrarily).

Phase 1 executes over steps $1$ to $T/2 - 2d'\ln n$ for a constant $d'>1$ of the current epoch $i$. Over this time, each ID $w$ generates random strings with associated outputs. After Phase 1 ends, IDs no longer generate new random strings.

Phase 2 begins at step  $T/2 - 2d'\ln n + 1$ and runs for $d'\ln n$ steps. Each ID $w$ (using its \qm~$\Q_w$) selects the string  $s^{\mbox{\tiny min}}_{w}$ with the smallest output $t^{\mbox{\tiny min}}_{w}$ that was generated in Phase 1, and then sends $s^{\mbox{\tiny min}}_{w}$  its neighbors. ID $w$ updates the corresponding bin and counter, as described earlier. 

Each neighbor $u$ verifies $s^{\mbox{\tiny min}}_{w}$.  Using $t^{\mbox{\tiny min}}_{w}$, ID $u$ decides whether to forward $s^{\mbox{\tiny min}}_{w}$ to its own neighbors (except for $w$) and, if so, updates the corresponding bin and counter; otherwise, $u$ ignores this value.  At the end of Phase 2, each ID $w$ selects the string with the smallest corresponding output it has seen so far; this is denoted by {\boldmath{$s^{i*}_w$}}.

Phase 3 starts at step $T/2 - d'\ln n + 1$ and runs for the final $d'\ln n$ steps. Over these steps, IDs no longer generate new strings, although they will still propagate them according to the above rules. 

At the end of the phase, each ID $w$ creates its solution set $R^w_i$ in the following way. ID $w$ finds the largest $j$ for which $B_j$ contains at least one element. Then, $w$ collects the corresponding string, and all the corresponding strings that have outputs contained in subsequent bins for decreasing $j$, until there are $d_0\ln n$ elements; the collection of these strings form $R^w_i$.


This concludes the propagation protocol. We note that immediately (at step $T/2 + 1$) ID $w$ will start generating a new ID signed with the string  $s^{i*}_w$ chosen in Phase 2.  
\smallskip\smallskip


\noindent{\bf Discussion.} The adversary may prevent good IDs from agreeing on the same solution set. As mentioned in Subsection~\ref{subsec:generatingID},  a $1/\textup{poly}(\log n)$-fraction may be unable to partake in the propagation process even with our secure routing, and their loss is already incorporated into our analysis in Subsection~\ref{sec:dynamic_case}. Therefore, we address the \defn{giant component} of $(1 - 1/\textup{poly}(\log n))n$ good IDs that can reach each other; going forward, we mean this set of IDs when referring to the ``good IDs''.

The critical source of disagreement between these good IDs is that the adversary may delay releasing a string $s'$ (or multiple strings) with a small output. For example, if this occurs right before the end of Phase 2, then only a subset of good IDs receive $s'$ and their respective solution sets differ from the other good IDs. 

We sketch how this disagreement is handled, but first we address the simpler case where there is no adversarial interference.\smallskip

\noindent{\it{\underline{With No Adversary}:}} Note that the propagation of a string in the giant component requires at most $d'\ln n$ steps. Therefore, since all IDs send their string at the beginning of Phase 1, then by the end of Phase 2, all IDs accept the same set of strings and agree on the minimum string. 

Furthermore, in Phase 3, nothing will occur (since no strings are released late) and so any IDs $w$ and $u$ are guaranteed w.h.p. to have $R^w_i = R^u_i$. What are the outputs corresponding to these solution sets? There are $\Theta(n)$ IDs computing for $\Theta(T)$ steps, so the smallest output in a set $R^w_i$ is $\Theta(\frac{1}{nT})$ and w.h.p. no larger than $O(\frac{\ln n}{nT})$.\smallskip

\noindent{\it{\underline{With an Adversary}:}} The adversary can propagate a string $s'$ with a small output late in Phase 2.\footnote{The adversary may also delay a string from a good ID outside the giant component, which amounts to the same problem since the adversary controls when this string is released into the giant component.} If $w$ receives $s'$ while $u$ does not, then $R^w_i \not= R^u_i$. We  argue that w.h.p. for good IDs in the giant component that (1) the size of each solution set remains bounded by $\Theta(\ln n)$, and (2) that the string ${\boldmath{s^{i*}_w}}$ used by each good ID $w$ belongs to every other good IDs' solution set; these two properties enable efficient and correct verification (described below). 

How many solutions $s'$ could $w$ receive and add to $R^w_i$? As noted above, this solution set will hold outputs of value $O(\frac{\ln n}{nT})$. Since the adversary has bounded computational power of $\beta n$, w.h.p. there cannot be more than $d''\ln n$ solutions with output value $\Theta(\frac{1}{nT})$ for some constant $d''>0$. This is true even if the adversary computes over the entire epoch. We set the constant $c_0$ used in the bin counters such that $c_0\geq d''$ in order to make sure that no smallest values are omitted. 

Now, consider two good IDs $w$ and $u$ in the giant component. Assume that $w$ selects $s^{i*}_w$ (recall that this occurs at the end of Phase 2), but that $s^{i*}_w$ is not present in good ID $u$'s solution set $R^u_i$ by the end of Phase 3; we will derive a contradiction.  If $s^{i*}_w$ originated from a good ID, then $u$ received $s^{i*}_w$ by the end of Phase 3 since $2d'\ln n$ steps is more then sufficient for the propagation of a string in the giant component (since this is the portion of the network to which we can search). Else, $s^{i*}_w$ originated from the adversary or a ID outside the giant component. Since $s^{i*}_w$ was held by $w$ by the end of Phase 2 (this inclusion could have been delayed by the adversary until the final step of Phase 2), the addition $d'\ln n$ steps in Phase 3 would have allowed $s^{i*}_w$ to reach $u$ and be added to $R^u_i$. In either case, this yields the contradiction.\smallskip\smallskip

\noindent{}Finally, what is the message complexity of the propagation protocol? Recall that for each bin, the associated  counter restricts to $O(\ln n)$ the number of times a ID forwards a string to its neighbors. Given that there are $O(\ln (nT))$ bins, the total number of times a ID can forward a string is $O(\ln(n)\,\ln (nT) )$. The number of messages sent between any pair of neighboring \qms~is $O(|\Q|^2)  = O( (\log\log n)^2)$ and the degree in the \qm~graph is $O(\textup{poly}(\log n))$. Therefore, the total message complexity over $O(n)$ IDs is $\tilde{O}(n \ln T )$ where $\tilde{O}$ accounts for $\textup{poly}(\log n)$ terms.\smallskip\smallskip

\noindent The above discussion supports the following:\smallskip\smallskip

\noindent{\bf \textsf{Lemma}}~{\bf \textsf{\ref{lem:propagate}}}. {\it With high probability, the protocol for propagating strings (i) guarantees that, for each good ID $w$ in the component, $s^{i*}_w$ is contained within the solution set of every good ID in the component, (ii) $|R^w_i| = O(\ln n)$, and (iii) has message complexity $\tilde{O}(n \ln T )$.
}

\smallskip\smallskip
\noindent{\bf Verifying IDs.} For simplicity, our discussion of ID verification in Subsection~\ref{subsec:generatingID} assumed that a single $r_{i-1}$ was agreed upon. However, not much changes when using solution sets. 

To generate an ID for use in epoch $i+1$, ID $w$ uses $s^{i*}_w$ to sign its ID.  By the above discussion, we are guaranteed w.h.p. that $s^{i*}_w$ belongs to the solution set of each good ID. Therefore, a good ID $u$ that wishes to verify $w$'s new ID checks whether this ID was signed by {\it any} of the strings in $R^u_i$; this requires checking only $O(\ln n)$ elements by the above discussion.


\section{Bootstrapping Groups}\label{sec:boot}

In prior work, an ID $v$ joins the network by contacting $O(\log n)$ members of a group which, by virtue of its size, has a good majority with high probability. In our construction, we make a similar assumption, except that $v$ contacts $O(\log n/\log\log n)$ groups selected uniformly at random, each of size $O(\log\log n)$. With high probability, the collection of the $O(\log n)$ IDs belonging to these groups will contain a good majority. Therefore, assuming that $v$'s ID is verified, these IDs can fulfill the role of a bootstrapping group $\Q_{\mbox{\tiny boot}}$, as described in Section~\ref{sec:dynamic}.  

During this bootstrapping process, we note that for input graphs with $O(1)$ expected degree, this approach will increase the expected state-maintenance cost to be $O(\log n)$. In the case of input graphs that have $O(\log n)$ degree, our expected state-maintenance cost remains unchanged.

\section{System Initialization}\label{app:initialization}

How are the group graphs $\mathcal{G}_1^0$ and $\mathcal{G}_2^0$ created? Almost all of the literature treats the problem of building secure overlays in the following manner: (1) prove that an overlay construction yields security guarantees, and then (2) prove these security guarantees can be maintained for each departure/join event.  This is typically challenging in its own right, and the issue of how the system is initially placed into a state that satisfies (1) is often not dealt with explicitly.

As motivation, one might consider a system whose development is shepherded by some central authority, or some small group of trustworthy participants, that handles admission control and the assignment of group membership. Once the system reaches some threshold size (sufficient for the w.h.p. arguments), the system becomes fully decentralized.

A notable exception is~\cite{guerraoui:highly}, which gives an explicit solution. The authors specify a protocol that allows all good IDs to learn about the existence of all good IDs; this incurs a communication cost of $O(n\cdot |E|)$ where $|E|$ is the number of edges in the overlay being created. Then, since all good IDs are now aware of each other, a subset of the IDs --- called a ``representative cluster'' --- is elected via running a Byzantine Agreement (BA) protocol; this has a  communication complexity of soft-$O(n^{3/2})$.   The representative cluster, which is shown to have an honest majority, is then responsible for establishing cluster membership, informing each ID about its fellow cluster members, and setting up links between the clusters. Therefore, the system initialization can be achieved w.h.p by this one-time ``heavy-weight'' procedure, after which the security guarantees can be maintained. 

We believe that a similar approach would work in our setting and allow for the creation of  $\mathcal{G}_1^0$ and $\mathcal{G}_2^0$. Since the problem of initialization pertains to much of the literature in this area, a pertinent question is whether one can improve upon the above scheme. We believe this is a problem in its own right and an interesting avenue for future work.\smallskip

\noindent{\bf Setting {\boldmath{$k$}}.} We end this section by addressing a subtlety in setting the value of $k$ sufficiently large; this was  referred to in Section~\ref{sec:static}.

As described in the main text, we assume that the neighbor sets in $\mathcal{G}_1^0$ and $\mathcal{G}_2^0$ are set up correctly for blue groups. Thus, $\pbad$ corresponds to the probability of group being bad, and we can set $k$ to be as large a constant as desired by setting $d_1$ sufficiently large. By Lemma~\ref{lem:dyn-aer}, $k\geq 2c+\gamma$ is sufficient. 

Now, consider the next epoch. The construction of the two new group graphs are proven to preserve $\pbad \leq 1/\log^k n$, where a red group corresponds to either a group that is bad {\it or} that has its neighbor set incorrectly established. Therefore, w.h.p., for subsequent epochs, groups graphs will be $\varepsilon$-robustness.

\end{document}